%% file: main.tex
\documentclass[runningheads,a4paper]{llncs}

\input{macros}

\begin{document}
\mainmatter  

\title{Quantitative Verification of Opacity Properties in Security Systems}
\titlerunning{Quantitative Verification of Opacity Properties in Security Systems}

 \author{Chunyan Mu\inst{1} \and David Clark\inst{2}} 
 \authorrunning{C. Mu \and D. Clark} 
 \institute{$^1$ Department of Computing and Games, Teesside University, U.K.\\
 	$^2$ Department of Computer Science, University College London, U.K.
 	\mailsa
 }

\maketitle

\begin{abstract}
 \noindent
 We delineate a methodology for the specification and verification of flow security properties expressible in the opacity framework. We propose a logic, \opacctl, for straightforwardly expressing such properties in systems that can be modelled as partially observable labelled transition systems.  We develop verification techniques for analysing property opacity with respect to observation notions. Adding a probabilistic operator to the specification language enables quantitative analysis and verification. This analysis is implemented as an extension to the PRISM model checker and illustrated via a number of examples. 
Finally, an alternative approach to quantifying the opacity property based on entropy is sketched.
 
\end{abstract}

\input{introduction}

\input{prelim}

\input{opac-ctl}

\input{opac-pctl}

\input{case-study}
\input{entropy}
\input{related-work}
\input{conclusion}

\bibliographystyle{splncs03}
\bibliography{BIB-opacLTL}

\input{appendix}

\end{document}

%% file: macros.tex
\usepackage{enumerate,amsmath,amssymb,graphicx}
\usepackage{url}
\usepackage{fancyvrb}
\usepackage{graphicx}
\usepackage{epsfig}
\usepackage{tikz,xspace,multicol}
\usepackage[boxruled,vlined]{algorithm2e}
\usepackage[normalem]{ulem}
\usepackage[colorlinks=true]{hyperref}%

\usetikzlibrary{positioning,shapes,arrows,automata}
\tikzset{ state/.style={draw,ellipse,initial text=} }

\tikzset{
  loop above right/.style={above right, out= 60, in= 30, loop},
  loop above left/.style ={above left,  out=150, in=120, loop},
  loop below right/.style={below right, out=330, in=300, loop},
  loop below left/.style ={below left,  out=240, in=210, loop}
}

\newcommand{\ie}    	{i.e., }
\newcommand{\eg}    	{e.g., }
\newcommand{\wrt}   	{w.r.t.\ }
\newcommand{\true}   	{{\mathtt{true}}}
\newcommand{\false}   	{{\mathtt{false}}}

\newcommand{\obs}   	{{\mathit{obs}}}
\newcommand{\prob}   	{{\sf {Prob}}}

\newcommand{\last}   	{\mathit{last}}

\newcommand{\dist}   	{{\sf{Dist}}}
\newcommand{\opacctl}	{{\sf OpacTL}\,}
\newcommand{\opacpctl}	{{\sf OpacPTL}\,}

\def\F{\mathbf{F}}
\def\P{\mathbf{P}}
\def\R{\mathbf{R}}

\def\D{\mathcal{D}}
\def\T{\mathcal{T}}
\def\PT{\mathcal{PT}}

\def\U{\mathbf{U}}

\def\X{\mathbf{X}}

\def \H{\sf {H}}
\def \L{\sf {L}}

\def \opac{\odot}

\newcommand{\TR}[1]{\mathsf{tr}(#1)}
\newcommand{\OPAC}[1]{\opac \lbrack #1 \rbrack}
\newcommand{\POWERSET}[1]{\frak{P}(#1)}
\newcommand{\RUN}[1]{\mathsf{path}(#1)}

\newcommand{\ERASE}[1]{\mathsf{erase}(#1)}
\newcommand{\ASP}{\mathsf{Asp}}
\newcommand{\SAT}{\mathsf{Sat}}
\newcommand{\POST}{\mathsf{Post}}
\newcommand{\PRE}{\mathsf{Pre}}

\newcommand{\FS}{\rightarrow}
\newcommand{\CAL}[1]{\mathcal{#1}}

\newcommand{\LLL}{\CAL{L}}
\newcommand{\MMM}{\CAL{M}}
\newcommand{\HHH}{\CAL{H}}
\newcommand{\PPP}{\mathbb{P}}
\newcommand{\NAT}{\mathbb{N}}

\newcommand{\TRANS}[1]{\stackrel{#1}{\longrightarrow}}
\newenvironment{GRAMMAR}{\[\begin{array}{lcl}}{\end{array}\]}
\newcommand{\VERTICAL}{\  \mid\hspace{-3.0pt}\mid \ }

\definecolor{blue}{rgb}{0,0,1}

\urldef{\mailsa}\path|c.mu@tees.ac.uk, david.clark@ucl.ac.uk|
\newcommand{\keywords}[1]{\par\addvspace\baselineskip
\noindent\keywordname\enspace\ignorespaces#1}

%% file: introduction.tex
\section{Introduction}
\label{sec:intro}

Security of software can be stated and investigated only with respect
to a specification of what security means in a given context. The
proliferation of software has lead to a profusion of security
properties, many of which exhibit a great deal of similarity. This has
lead to a search for common abstractions, and general formalisms for
expressing security properties.  \emph{Opacity}, first introduced by Mazar\'{e} 
\cite{MazareL:usiunifop}, is a promising approach for describing and
unifying security properties.  The key insight behind opacity is
twofold:

\begin{itemize}

\item Distinguish between system behaviour, and observations of the
  system's behaviour.

\item A security property $\varphi$ is opaque, provided that for every
  behaviour $\pi$ satisfying $\varphi$ there is another behaviour
  $\pi'$, not satisfying $\varphi$, such that $\pi$ and $\pi'$ give
  rise to identical observations.

\end{itemize}
Clearly, when $\varphi$ is opaque in this sense, observers cannot be
sure if $\varphi$ holds of the observed system, or not.  
Many familiar security definitions such as non-interference, declassification, and 
knowledge-based security can be obtained by suitable
choices of predicate $\varphi$~\cite{SchoepeS15}, 
and notion of observable behaviour.

Absolute guarantees of security are often impractical, and we may
tolerate violation of security properties, as long as it happens
with sufficiently low probability.  So verifying and analysing
security properties can also usefully take quantitative aspects into
account.  For this reason, opacity has been extended to
quantitative opacity~\cite{BerardMS10,BryansKM12}.  

While opacity successfully unifies many concrete notions of software
security, it is a semantic concept, typically described 
informally using the mathematical vernacular. There are not yet dedicated development of tools for automatic verification of opacity
properties. The present paper asks several questions:

\begin{itemize}

\item Can we build a sufficiently expressive \emph{logic} for opaque
  predicates?

\item Can we build automated verification tools for opaque predicates?

\item Can we extend the logic and automated verification for opacity to
  probabilistic opacity?

\end{itemize}
We answer all questions in the affirmative.

First, we propose a logic $\opacctl$ for expressing opacity
and internalise its definition in the logic: 
for this purpose we extend conventional temporal logic with a
predicate
\[
   \OPAC{\psi}
\]
which expresses that property $\psi$ is opaque.  
$\opacctl$ allows us to specify the more general opacity property in a \emph{straightforward} way, 
and the property required to be secret can be defined flexibly regarding users' requirements. 
We also present $\opacpctl$, a probabilistic generalisation of $\opacctl$, 
which enables us to express probabilistic opacity,
and allows us to reason about the \emph{degree} of satisfaction or violation of the security property of interest.
We demonstrate the expressibility of both $\opacctl$ and $\opacpctl$
through examples of opacity, respectively probabilistic opacity, from the literature. 
We automate both logics by translation to the PRISM model checker.
 We also introduce an
 alternative approach to measure the level of opacity of the system, based on the notion of entropy.

 \paragraph{Outline.} This paper is organised as follows.
 In Section~\ref{sec:prelim} we recall the definition of labelled
 transition system, observations, and opacity.
 Section~\ref{sec:opacctl} introduces a temporal logic \opacctl\, for
 the formal specification of opacity.  
 Section~\ref{sec:opacpctl} presents the probabilistic
 extension of the logic \opacpctl, and a prototype model checker for
 the (probabilistic) opacity fragment of the logic.
 Section~\ref{sec:eg} describes the implementation of our techniques
 for analysing (quantitative) opacity flow properties of systems in a
 prototype tool, and demonstrates its applicability using several case
 studies.  
Section~\ref{sec:entropy} provides an alternative measurement of the opacity formula based on the notion of \emph{entropy}.
 Section~\ref{sec:related-work} briefly reviews  literature
 in related areas,
 and Section~\ref{sec:conc} draws conclusions
and points out some future directions.

%% file: prelim.tex
\section{Labelled transition systems and observations}
\label{sec:prelim}

Let $\NAT$ be the set of natural numbers, assume $0 \in \NAT$.
An \emph{alphabet} $\Sigma$ is a non-empty, finite set,
$|\Sigma|$ is its cardinality. 
$\Sigma^*$ denotes the set of all \emph{finite} words over $\Sigma$ 
including the \emph{empty word $\varepsilon$},
$\Sigma^n$ denotes the set of all \emph{finite} words over $\Sigma$ and the length of the words is $n$,
 $\Sigma^+ = \Sigma^* \setminus \{\varepsilon\}$,
$\Sigma^{\omega}$ denotes the set of all \emph{infinite} words, $\Sigma^{\infty}$ denotes the set of all \emph{finite} and \emph{infinite} words.
The subsets of $\Sigma^*$:$ L \subseteq \Sigma^*$ are called languages,
and $L \subseteq \Sigma^{\infty}$ are called $\omega$-languages.

\begin{definition}\label{def_lts}
A \emph{labelled transition system} (LTS) is a tuple $\LLL = (S, \Sigma, \FS, F)$, 
where: 
\begin{itemize}

    \item $\Sigma$ is a finite alphabet of \emph{labels}, or
      \emph{actions}, including a distinguished element $\bot$
      representing termination (staying there forever).

    \item $S$ is a finite set of \emph{states}.

    \item $\FS \subseteq S \times \Sigma \times S$ is the
      \emph{transition relation}.

	\item $F \subset S$ is the set of final states, a (possibly empty) subset of $S$.
\end{itemize}
\end{definition}

\begin{definition}
We say $\LLL$ is \emph{deterministic} iff $\forall s,t \in S$, whenever  $s \TRANS{a} t$ and
$s \TRANS{a} t'$ then $t = t'$.  We say $\LLL$ is \emph{circular},
if each state has an outgoing transition, \ie for all $s \in S$ there
is $s \TRANS{a} t$.  A \emph{path in $\LLL$} is a total map $\pi :
\NAT \FS (\Sigma \cup S)$, subject to the following constraints:
\begin{itemize}

\item Whenever $i \in \NAT$ is even, then $\pi(i)$ is a
  state. Otherwise $\pi(i)$ is a label.

\item For all even $i$ we have $\pi(i) \TRANS{\pi(i+1)} \pi(i+2)$.

\end{itemize}
The set of all $\LLL$'s paths is denoted by $\RUN{\LLL}$.  The set of paths of
$\LLL$ starting from state $s$ is denoted by $\RUN{\LLL,s}$.  We write $\pi \lbrack i \dots \rbrack$ to denote $\pi(i) \pi(i+1) \dots$, 
and write $\ERASE{\pi}$ for the map that erases all the states from $\pi$: $\ERASE{\pi} = \pi(1)\pi(3)\dots$, in other words, 
$n \mapsto \pi(2n+1)$, defined for all $n \in \NAT$.  
A \emph{trace in $\LLL$} is a map $tr : \NAT \FS \Sigma$, 
such that $tr = \ERASE{\pi}$ for some path $\pi$.
\end{definition}
We are often sloppy when talking about paths and traces, in particular,
we often elide the indices, writing e.g.~$ab\bot \bot...$ for a trace
$\{(0, a), (1, b)\} \cup \{(n, \bot)\ |\ n > 1\}$.

\begin{definition}
A trace $tr$ is \emph{well-structured} iff for all suitable $i < j$ we
have: $tr(i) = \bot \Rightarrow tr(j) =\bot$. Such the smallest $i$ is denoted by
$\last$.  A path $\pi$ is \emph{well-structured} iff $\ERASE{\pi}$ is
well-structured, $\pi(\last*2)$ is called a \emph{final state}. 
$\LLL$ is \emph{well-structured} iff all of its paths
are well-structured.  Traces and paths are \emph{semantically finite}
iff they are well-structured, and at least one of their labels is
$\bot$.
\end{definition}
Note that a well-structured LTS won't change a state after a termination. A key idea in modelling security properties is the observation power of the attacker. 
We use a set of \emph{observables}, distinct from the states and actions of the
 LTS, for this purpose.  Actions, states and observables are connected
by an observation function.

\begin{definition}
Let $\Theta$ be a finite alphabet for observables. We write
$\Theta_{\bot}$ for $\Theta \cup \{\bot\}$, assuming that $\bot \notin
\Theta$.  A \emph{observation function} is a function 
$\obs: \RUN{\LLL} \to \Theta^{*}_{\bot}$, 
subject to the additional constraint that $\obs(\bot)=\bot$: 
\ie for all paths $\pi$ of the form $\pi = \pi_L\bot\pi_R$ 
we have $obs(\pi) = obs(\pi_L)\bot obs(\pi_R)$.
Observation functions on traces are defined similarly.
\end{definition}
The intuition is that $\Theta$ contains all possible projections of paths. 
In practice many observation functions will have additional structure, 
\eg $\obs(\pi) = \obs'(\pi(1)) \obs'(\pi(3)) \obs'(\pi(5))\dots$, 
for some function $\obs': \Sigma \to \Theta_{\bot}$ on transition labels. 

Opacity is a general framework for formulating and unifying security
properties expressed as predicates. Here, a predicate is simply a subset of $\RUN{\LLL}$.
A predicate is \emph{opaque} if, given any path of the system, an adversary's observation 
is unable to determine whether the path satisfies the predicate.
In comparison with other security policies such as non-interference, 
the opacity framework allows a more flexible specification 
of both the adversary's power of observation and the confidential properties of the system.
%
%
We use $\varphi$ as a shorthand for a set of traces satisfying $\varphi$ in the paper somewhere.
\begin{definition}{[Opacity and observability]}\label{defn_opacity}
  A \emph{predicate} $\varphi$ over $ \RUN{\LLL}$ is a subset
  of $ \RUN{\LLL}$.  Given an observation function $\obs$, a predicate $\varphi$ is
  \emph{opaque \wrt $\obs$} iff: for every path $\pi \in \varphi$, there
  is a path $\pi' \not\in \varphi$ such that $\obs(\pi) = \obs(\pi')$,
  \ie all paths in $\varphi$ are covered (observationally equivalent) by paths in $\bar\varphi$ 
  (the complement of $\varphi$):
  $\obs(\varphi) \subseteq \obs(\bar \varphi)$.
  Paths in $\varphi$ are called \emph{observable} paths 
    iff there is at least one path in $\varphi$ which is not covered by paths in $\bar\varphi$: $\varphi \setminus \obs^{-1}(\obs(\bar \varphi)) \neq \emptyset$.
\end{definition}

Many special cases of opacity are easily definable in our approach,
such as \emph{initial-state opacity}, \emph{final-state opacity},
\emph{initial-final-state opacity}~\cite{WuL13},
\emph{total-opacity}~\cite{BryansKMR08} and \emph{language-based
  opacity}~\cite{SabooriH13,SabooriH14}. In Section~\ref{sec:opacpctl}
we will generalise opacity to \emph{probabilistic opacity}~\cite{BryansKM12}.





%% file: opac-ctl.tex
\section{The Opaque Temporal Logic \opacctl}
\label{sec:opacctl}
The behaviour of a state transition system is described as sequences
of labels during the possible executions.
The labels indicate the valuations of the input/output variables of the system.
The temporal properties we specify are upon such behaviours over the same alphabet.
We study here \opacctl, which is a temporal logic~\cite{ClarkeE81} 
with an opacity operator $\opac$ over \emph{semantically finite} traces for specification of security properties. 

\subsection{Formulae}
Let $\ASP$ be a fixed set, the \emph{atomic state propositions},
ranged over by $\alpha$.  We have two classes of formulae given
by the grammar below: \emph{state formulae} and
\emph{path formulae} ranged over by $\phi$ and $\psi$ respectively.

\begin{GRAMMAR}
  \phi
     &\qquad::=\qquad&
  \true
     \VERTICAL
  \false
     \VERTICAL
  \alpha
     \VERTICAL
  \neg \phi
     \VERTICAL
  \phi \land \phi
     \VERTICAL
  \OPAC{\psi}
     \\
  \psi
     &\qquad::=\qquad&
  \X \phi 
     \VERTICAL
  \phi  \U \phi 
     \VERTICAL
     \phi  \R \phi
     \VERTICAL
  \bot
    \VERTICAL
  \neg \psi
\end{GRAMMAR}
Note that an \opacctl formula is defined relative to a state and always a state formula.
Path formulae only appear inside the $\OPAC{\cdot}$ operator.

\subsection{Model}\label{def_model}

A \emph{model} is a tuple $(\LLL, \eta, \obs)$ where:
\begin{itemize}
  
\item $\LLL = (\Sigma, S, \FS, F)$ is an LTS, that is circular,
  deterministic and well-structured.
  
\item  $\eta : S \FS \POWERSET{\ASP}$ is the state labelling function,
where $\POWERSET{\ASP}$ is the powerset of $\ASP$.

\item $\obs : \RUN{\LLL} \to \Theta^{*}_{\bot}$ is an observation
  function.

\end{itemize}

\subsection{Satisfaction relations}
As we have two notions of formulae (state and path
formulae), we need two satisfaction relations.  Let $\MMM = (\LLL,
\eta, obs)$ and $s \in S$.  We now define the \emph{satisfaction relation}
$\MMM \models_s \phi$ for state formulae.

\begin{itemize}

  \item $\MMM \models_s \true$ always holds.

  \item $\MMM \models_s \alpha$ iff $\alpha \in \eta(s)$.

  \item $\MMM \models_s \neg\phi$ iff $\MMM \not \models_s \phi$.

  \item $\MMM \models_s \phi \land \phi'$ iff $\MMM \models_s \phi$ and  $\MMM \models_s  \phi'$.

 \item $\MMM \models_s \OPAC{\psi}$ iff for all paths $\pi \in  \RUN{\LLL, s}$
     whenever $\MMM \models_{\pi} \psi$
     then there
     exists a path  $\pi' \in \RUN{\LLL, s}$ s.t.
     $\MMM \not\models_{\pi'} \psi$ and
     $\obs(\pi) = \obs(\pi')$.
\end{itemize}

\noindent Now let $\pi$ be a path in $\LLL$. Then we define:

\begin{itemize} 

\item $\MMM \models_{\pi} \X \phi$ iff $\MMM \models_{\pi(2)} \phi$.
  
\item $\MMM \models_{\pi} \phi \U \phi'$ iff $\exists i \in \NAT.\MMM \models_{\pi(2i)} \phi'$ and 
  $\forall 0 \le j < i.\MMM \models_{\pi(2j)} \phi$.

\item $\MMM \models_{\pi} \phi \R \phi'$ iff 
     either $\exists i \in \NAT.\MMM \models_{\pi(2i)} \phi \land \forall 0 \le j \le i.\MMM \models_{\pi(2j)} \phi'$ 
      or $\forall j \in \NAT.\MMM \models_{\pi(2j)} \phi'$.

\item $\MMM \models_{\pi} \bot$ iff $\pi(1) = \bot$.
  
\item $\MMM \models_{\pi} \neg\psi$ iff $\MMM \not \models_{\pi} \psi$.

  
\end{itemize}

%
%
%
%
%
%

\noindent Note that the opacity operator is general and can be used
to express initial-opacity (only $\pi(0)$ is sensitive), final-opacity
(only $\pi(\last*2)$ is sensitive (focus of this paper), and language-opacity ($\pi(i)$ is sensitive for all $i$). 


\begin{theorem}
Given a model $\MMM = (\LLL, \eta, \obs)$ and a state $s$ in
$\LLL$. Let $\psi$ be a path formula.
   Define: $
      \TR{\psi,s}
          =
      \{\ERASE{\pi}: \pi \in \RUN{\LLL, s}\ |\ \MMM \models_{\pi} \psi \}
   $,
   then we have:
   \[
   \MMM \models_s \OPAC{\psi}
      \qquad\text{iff}\qquad
   \TR{\psi,s}\ \text{is semantically opaque \wrt} \obs. 
   \]
   Here, ``semantically opaque'' is to be understood as the definition of opaque in
   Definition \ref{defn_opacity}.
\end{theorem}
\begin{proof}
According to the satisfaction relations, we have:
\begin{eqnarray*}
\MMM \models_{s} \OPAC{\psi} 
& \Leftrightarrow &
 \forall \pi \in \RUN{\LLL,s}. (\MMM \models_{\pi} \psi 
 ~ \Rightarrow ~ \\
&& \exists \pi' \in \RUN{\LLL,s} 
~ s.t.~ (\MMM \not\models_{\pi'} \psi 
		\land \obs(\pi) = \obs(\pi')  ) )\\
& \Leftrightarrow &
\{\obs(\pi)  \mid \MMM \models_{\pi} \psi\}  
\subseteq 
\{\obs(\pi') \mid \MMM \models_{\pi'} \neg\psi\} \\
& \Leftrightarrow &
\{ \pi \in \RUN{\LLL, s} \mid \MMM \models_{\pi} \psi \}\ \text{are covered by paths violating $\psi$.} \\
& \Leftrightarrow &
\TR{\psi,s}\ \text{is semantically opaque.} 
\qquad \qquad \qquad \qquad \qquad \qquad \ \ \ \hfill\Box
\end{eqnarray*}
\end{proof}
The opacity computation problem is the problem of recognising whether there is a path violating $\psi$ but observationally equivalent to each of the $\psi$ satisfying paths.
%

\subsection{Verification of $\opacctl$}
\label{subsec:veri-opacctl}
The verification problem of $\opacctl$ is a decision algorithm to check whether $\MMM \models_s \phi$, for a given model $\MMM$, and an $\opacctl$ formula $\phi$ and a starting state $s$. That is, we need to  set up whether the formula $\phi$ is valid in the initial state $s$ of $\MMM$. Let $\POST(s)$ denote immediate state successors of $s$ in a path, and $\PRE(s)$ denote the immediate state predecessors of $s$ in a path. The basic procedure follows the conventional CTL model checking~\cite{Baier2008}: 

\begin{itemize}
   \item [(i)] Convert the $\opacctl$ formulae in a positive normal form, that is, formulae built by the basic modalities $\OPAC{\X \phi}$, $\OPAC{\phi  \U \phi'}$, and $\OPAC{\phi \R \phi'}$, and successively pushing negations inside the formula at hand: $\neg \true \leadsto \false$, $\neg \false \leadsto \true$, $\neg\neg \phi \leadsto \phi$, $\neg(\phi \land \phi') \leadsto \neg \phi \lor \neg \phi'$, $\neg(\phi \lor \phi') \leadsto \neg \phi \land \neg \phi'$, $\neg \X \phi \leadsto \X \neg \phi$, $\neg(\phi \U \phi') \leadsto \neg\phi \R \neg\phi'$, $\neg(\phi \R \phi') \leadsto \neg\phi \U \neg\phi'$;
   \item [(ii)] Recursively compute the satisfaction sets $\SAT(\phi') = \{s \in S \mid s \models \phi'\}$ for all state subformulae $\phi'$ of $\phi$: the computation carries out a bottom-up traversal of the parse tree of the state formula $\phi$ starting from the leafs of the parse tree and completing at the root of the tree which corresponds to $\phi$, where the nodes of the parse tree represent the subformulae of $\phi$ and the leafs represent an atomic proposition $\alpha \in \ASP$ or $\true$ or $\false$. 
All inner nodes are labelled with an operator. For positive normal form formulae, the labels of the inner nodes are $\neg$, $\land$, $\OPAC{\X}$, $\OPAC{\U}$, $\OPAC{\R}$.
At each inner node, the results of the computations of its children are used and combined to build the states of its associated subformula. 
In particular, satisfaction sets for conventional state formula are given as follows: 
   \begin{itemize}
	 \item $\SAT(\true) = S$,
	 
     \item $\SAT(\alpha) = \{t \in S \mid \alpha \in \eta(t)\}$,

     \item $\SAT(\neg \phi) = S \setminus \SAT(\phi)$,

     \item $\SAT(\phi \land \phi') = \SAT(\phi) \cap \SAT(\phi')$,
     
     
     \item $\SAT(\OPAC{\psi}) = \{s \in S \mid \T_{\opac}(\MMM,s,\psi) = \true \}$;

     \end{itemize}
   \item [(iii)] Check whether $s \in \SAT(\phi)$.
\end{itemize}

Furthermore, for the opacity operator $\OPAC{\psi}$, 
we compute $\SAT(\OPAC{\psi})$ as: $s \in \SAT(\OPAC{\psi})$ iff $\T_{\opac}(\MMM,s,\psi) = \true$, $\T_{\opac}(\MMM,s,\psi) $ is sketched in Algorithm~\ref{algo:opac}. 
Specifically, we compute all (regular-expression-like formatted) traces 
$\Lambda$ ($\Lambda'$) starting from $s$ and satisfying (violating) $\psi$, and check if each such trace in $\Lambda$ is observationally covered by a such trace in $\Lambda'$. 
Alogrithm~\ref{algo:compU} $\textbf{compU}(\MMM, s, \phi, \phi')$ computes a set of such formatted traces satisfying $\phi \U \phi'$.
Similarly, an algorithm $\textbf{compR}(\MMM, s, \phi, \phi')$ can be proposed to compute a set of such formatted traces satisfying $\phi \R \phi'$.

\input{algo-opac}
\input{algo-compu}

\begin{theorem}\label{theo:soundness-opacctl}
[Soundness of $\OPAC{\psi}$ translation]
Given a model $\MMM$, a state $s$, and a path formula $\psi$:
\[\MMM \models_s \OPAC{\psi} \quad \text{iff} \quad \T_{\opac}(\MMM,s,\psi) = \true.\]
\end{theorem}
\begin{proof}
The proof is obtained by the satisfaction relation of $\OPAC{\psi}$
and the construction of the translation $\T_{\opac}(\MMM,s,\psi)$. Let $\pi$ (and $\pi'$) denotes a corresponding path of the regular-expression-like formatted trace $\lambda$ (and $\lambda'$ respectively).
\begin{eqnarray*}
\T_{\opac}(\MMM,s,\psi) = \true 
& \Leftrightarrow &
 \forall \lambda \in \SAT(\psi): 
 \exists \lambda' \in \SAT(\neg\psi). 
 \obs(\lambda) \subseteq \obs(\lambda') \\
& \Leftrightarrow &
 \forall \pi. \ERASE{\pi} \in \lambda \in \SAT(\psi): \\
 &&
 \exists \pi'. \ERASE{\pi'} \in \lambda' \in \SAT(\neg\psi). 
 \obs(\pi) = \obs(\pi') \\
 & \Leftrightarrow &
 \forall \pi \in \RUN{\LLL,s}.\MMM_{\pi} \models \psi: \\
 &&
 \exists \pi' \in \RUN{\LLL,s}.\MMM_{\pi} \not\models\psi.(\obs(\pi) = \obs(\pi')) \\
 & \Leftrightarrow &
 \MMM \models_s \OPAC{\psi}.
 \qquad\qquad\qquad\qquad\qquad\qquad\qquad\qquad\qquad\qquad \Box
\end{eqnarray*}
\end{proof}

Due to the recursive nature of the CTL model checking algorithm,
complexity is linear in the size of the non-opacity state formula $\phi$,
while the worst case of finding opaque paths for reachability objectives $\psi$ 
and checking satisfaction of the opacity state formula $\OPAC{\psi}$, 
specified in Algorithm~\ref{algo:opac}, 
is EXPSPACE.
Note that the algorithm traverses all traces satisfying $\psi$ and all traces violating $\psi$, and conducts observation equivalence comparison.
So the worst case complexity here follows the complexity of the hyper property model checking problem with two quantifier ($\forall$) alternations, and thus EXPSPACE~\cite{BonakdarpourFinkbeiner21}.




The until operator allows to derive the temporal modality $\F$ (``eventually'') as usual: $\F\phi \ {\buildrel\rm def\over=} \ \true~ \U~ \phi$. To simplify expression in the examples, by abuse of notation, we use $\F s$ to denote $\F \phi$ (\ie $\phi$ will be eventually true which takes place on state $s$).

\input{eg-opacctl}

\input{link-ni}

%% file: algo-opac.tex
\small{\begin{algorithm}[h!]
 \SetAlgoLined
  \KwData{$\MMM, s, \psi$}
  \KwResult{$\OPAC{\psi}$}
  \Switch{$\psi$} {
	\textbf{case} $\X \phi$: \ \ \ $\SAT(\psi) \leftarrow \cup_{i \in \NAT}\{tr(s \to s_i) \mid \POST(s)=s_i  \land s_i \in \SAT(\phi)\}$, \\ 
	$\qquad \qquad \quad \ \SAT(\neg\psi) \leftarrow \cup_{i \in \NAT}\{tr(s \to s_i) \mid \POST(s)=s_i  \land s_i \in \SAT(\neg\phi)\}$\;
	\textbf{case} $\phi \U \phi'$: $\SAT(\psi) \leftarrow \textbf{compU}(\MMM, s, \phi, \phi')$, \\ 
	$\qquad \qquad \quad \ \SAT(\neg\psi) \leftarrow \textbf{compR}(\MMM, s, \neg\phi, \neg\phi')$\;
	\textbf{case} $\phi \R \phi'$: $\SAT(\psi) \leftarrow \textbf{compR}(\MMM, s, \phi, \phi')$, \\ 
	$\qquad \qquad \quad \ \SAT(\neg\psi) \leftarrow \textbf{compU}(\MMM, s, \neg\phi, \neg\phi')$\;
	}  
  	$\Lambda \leftarrow  \{\lambda \mid \lambda \in \SAT(\psi)\}$;
  	$\Lambda' \leftarrow \{\lambda \mid \lambda \in \SAT(\neg\psi)\}$\;
	\For{each $\lambda \in \Lambda$}
	{
		$\mathtt{find} \leftarrow \false$\;
		\For {each $\lambda' \in \Lambda'$}
		{
	  	\If{$\obs(\lambda) \subseteq \obs(\lambda')$}
	  	{
	    	$\mathtt{find} \leftarrow \true$; break \;
	  	}
	  }
	  \textbf{if } $(\lnot \mathtt{find})$ \ \textbf{then } \ \Return{$\mathtt{find}$}\;
	 }
  \Return{$\mathtt{find}$}.
  \caption{Translating $\OPAC{\psi}$: $\T_{\opac}(\MMM,s,\psi)$}
 \label{algo:opac}
\end{algorithm}}

%% file: algo-compu.tex
\small{\begin{algorithm}[h!]
 \SetAlgoLined
  \KwData{$\MMM, s, \phi, \phi'$}
  \KwResult{$\SAT(\phi \U \phi')$: a set of regular-expression-like formatted traces whose corresponding paths satisfying $\phi \U \phi'$}
    $\Lambda \leftarrow \{\}$; \ $i \leftarrow 0$ \;
	\For{each $t_i \in \SAT(\phi')$}
	{
	    $T_i \leftarrow \{t_i\}; \ \Pi_i \leftarrow \{\pi \mid \pi(0)=t_i\}$\;
	    \While{$\{s_j \in \SAT(\phi) \setminus (T_i \cup \SAT(\phi')) \mid \POST(s_j) \cap T_i \ne \emptyset\} \ne \emptyset$}
	    {
	    	let $s_j \in \{s_j \in \SAT(\phi) \setminus (T_i \cup \SAT(\phi')) \mid \POST(s_j) \cap T_i \ne \emptyset\}$\;
	    	\If{$s_j \in \POST(s_j) \cap T_i$}
	    	{
	    		/* There is a self-loop: wrap it with a star and concatenate paths starting from a state in $\POST(s_j) \cap T_i$ found earlier */ \\
	    		\For {each $\pi' \in \Pi_i$ s.t. $\pi'(0) \in \POST(s_j) \cap T_i$}
	    		{
	    			$\Pi_i \leftarrow \Pi_i \cup \{(s_j \TRANS{a} s_j)^* + \pi'\lbrack 1... \rbrack \}$\;
	    		}
	    	}	    
	    	\For {each: $q_1\in \POST(s_j) \cap T_i, q_2 \in \POST(q_1) \cap T_i, \dots, q_n \in \POST(q_{n-1}) \cap T_i$ s.t. $\POST(q_{n}) \cap T_i = \emptyset$}
	    	{
	    		\If{$\PRE(s_j) \not\in \{q_1, q_2, \dots q_n\}$}
	    		{
	    			\For {each $\pi' \in \Pi_i$ s.t. $\pi'(0) \in \POST(s_j) \cap T_i$}
	    			{
	    			$\Pi_i \leftarrow \Pi_i \cup \{s_j \TRANS{a} q_1 + \pi'\lbrack 1... \rbrack \}$\;
	    			}
	    		}
	    		\ElseIf {$\PRE(s_j) = q_n \land s_j \in \POST(q_{n})$}
	    		{
	    			/* There is a cycle, wrap it with a star and concatenate paths starting from a state in $\POST(s_j) \cap T_i$ found earlier */ \\
	    		
	    			\For {each $\pi' \in \Pi_i$ s.t. $\pi'(0) \in \POST(s_j) \cap T_i$}
	    			{
	    			$\Pi_i \leftarrow \Pi_i \cup \{(s_j \TRANS{a_1} q_1 \TRANS{a_2} \dots \TRANS{a_{n}} \PRE(s_j) \TRANS{a_{n+1}} s_j)^* + \pi'\lbrack 1... \rbrack\}$\;
	    			}
	    		}
	    	}
	    	$T_i \leftarrow T_i \cup \{s_j\}$\;
	    	}
	    $\Lambda_i = \{\lambda \leftarrow \ERASE{\pi} \mid \forall \pi \in \Pi_i \}$;  \
	    $\Lambda \leftarrow\Lambda \cup \Lambda_i$; \ 
	    $i \leftarrow i+1$\;
	 }
  \Return{$\Lambda$}.
  \caption{Computing $\SAT(\phi \U \phi')$: \bf{compU}($\MMM, s, \phi ,\phi'$)}
 \label{algo:compU}
\end{algorithm}}

%% file: eg-opacctl.tex
\begin{example}
\label{eg:opctl-1}
Consider the system $\MMM$ accepting finite inputs presented in Fig.~\ref{fig:eg-opacctl} (a).
Let $s_3$ be a sensitive state, $s_0$ be a starting state, 
$\{s_3, s_6\}$ be two final states of interests,
and the observation function be: 
$a \to a,\, b\to \epsilon,\, c\to \epsilon$, i.e. $b, c$ are hidden, but $a$ is visible.
Consider the security property \emph{eventually reaching $s_3$}, 
\ie $\psi = \F s_3$ in our logic. The operator $\F\phi$ can be defined as $\true~ \U~ \phi $, so:
\[ 
\TR{\psi,s_0} = \{ac(b)^*a(\bot)\!^{*}\}, 
\TR{\neg \psi,s_0} = \{aba(c)^*(\bot)\!^{*}\}. 
\]
Here we use e.g.~$ac(b)^*a(\bot)\!^{*}$ to represent the infinite trace that starts with
$ac$, followed by possibly infinitely many $b$, followed by a, and followed by possibly infinitely many $\bot$.  Clearly:
\[\obs(\psi) = \obs(\neg \psi) = \{aa(\bot)\!^{*}\},\] 
so: $\MMM \models_{s_0} \opac \lbrack
\psi \rbrack$, \ie the system is $\psi$-opaque.
\end{example}
\begin{example}
\label{eg:opctl-2}
Consider the system accepting infinite inputs presented in Fig.~\ref{fig:eg-opacctl} (b).
Let $s_2$ be a sensitive state, $s_0$ be a starting state, 
$\{s_2, s_6\}$ be two final states of interests,
and the observation function be: 
$a \to \epsilon,\, b\to b,\, c\to \epsilon$, i.e.~$a$ and $c$ are hidden, but $b$ isn't.
Consider the security property \emph{eventually reaching $s_2$}, 
\ie $\psi = \F s_2$:
\[\TR{\psi,s_0} = \{ab(ab)^{*}\},\qquad
\TR{\neg \psi,s_0} = \{acbc(bc)^{*}\}\]
clearly, $\obs(\psi) = \obs(\neg \psi) = \{bb^{*}\}$, 
so the system is $\psi$-opaque, \ie
$
\MMM \models_{s_0} \OPAC{\psi}.
$

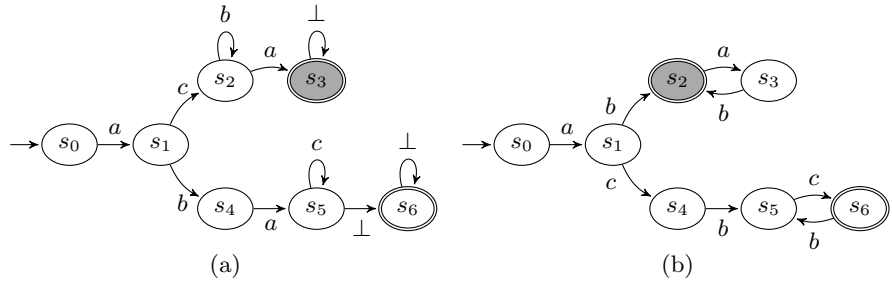
\begin{figure}[!h]
\centering
  \begin{tikzpicture}[->,>=stealth',shorten >=1pt,auto,node distance=1.2cm, scale = 1, transform shape]

  \node[initial,state] (A) {$s_0$};
  \node[state] (B) [right of=A] {$s_1$};
  \node[state] (C) [above right of=B] {$s_2$};
  \node[state, accepting,fill={rgb:black,1;white,2}] (D) [right of=C] {$s_3$};
  \node[state] (E) [below right of=B] {$s_4$};
  \node[state] (F) [right of=E] {$s_5$};
  \node[state, accepting] (G) [right of=F] {$s_6$};
  \node[align = center, below =2mm of E] {(a)};

   \path[every node/.style={font=\sffamily\small,
   		fill=white,inner sep=1pt}]
     (A) edge node [above=1mm]{$a$} (B)
     (B) edge [bend left=20] node[above=1mm] {$c$} (C)
	 				 edge [bend right=20] node[below=1mm] {$b$} (E)
     (C) edge [loop above] node[above=1mm] {$b$} (C)
     			 edge [bend left=20] node [above=1mm]{$a$} (D)
     (D) edge [loop above] node[above=1mm] {$\bot$} (D)
     (E) edge node [below=1mm]{$a$} (F)
     (F) edge [loop above] node[above=1mm] {$c$} (F)
     			 edge node [below=1mm]{$\bot$} (G)
     (G) edge [loop above] node[above=1mm] {$\bot$} (G)
;
\end{tikzpicture}
 \begin{tikzpicture}[->,>=stealth',shorten >=1pt,auto,node distance=1.2cm, scale = 1, transform shape]

  \node[initial,state] (A) {$s_0$};
  \node[state] (B) [right of=A] {$s_1$};
  \node[state, accepting,fill={rgb:black,1;white,2}] (C) [above right of=B] {$s_2$};
  \node[state] (D) [right of=C] {$s_3$};
  \node[state] (E) [below right of=B] {$s_4$};
  \node[state] (F) [right of=E] {$s_5$};
  \node[state, accepting] (G) [right of=F] {$s_6$};
  \node[align = center, below =2mm of E] {(b)};

   \path[every node/.style={font=\sffamily\small,
   		fill=white,inner sep=1pt}]
     (A) edge node [above=1mm]{$a$} (B)
     (B) edge [bend left=20] node[left=2mm] {$b$} (C)
	 				 edge [bend right=20] node[left=2mm] {$c$} (E)
     (C) edge [bend left=20] node [above=1mm]{$a$} (D)
     (D) edge [bend left=20] node [below=1mm]{$b$} (C)
     (E) edge node [below=1mm]{$b$} (F)
     (F) edge [bend left=20] node [above=1mm]{$c$} (G)
     (G) edge [bend left=20] node [below=1mm]{$b$} (F)
;
\end{tikzpicture}
\caption{\opacctl over finite and infinite languages. Grey node denotes a sensitive state.}
\label{fig:eg-opacctl}
\end{figure}
\end{example}

%% file: link-ni.tex
\subsection{Link to non-interference}
Information flow policies are designed to ensure that secret data does not influence publicly observable data.
Non-interference (NI)~\cite{GoguenMeseguer82} essentially says that low security users should not be aware of the activity of high security users. This can be interpreted in many different ways including as a hyperproperty of paths of a system. Appropriate to an LTS is to consider a system with labels classified into two security levels: \emph{low} ($\L$) and \emph{high} ($\H$). The system satisfies non-interference if and only if any sequence of low-level labels that can be produced by the LTS is consistent with all possible sequences of high-level labels. We formalise this notion of NI.

Let $trace(\LLL)$ be the set of all traces produced by paths of $\LLL$. We consider projections of these traces. Let $\LLL$ be an LTS over an alphabet $\Sigma$ which is partitioned into two disjoint sets of labels $\H$ and $\L$.  A \emph{high projection} of a trace $\lambda$ is the sequence of labels  in $\H$ that remains after all labels in $\L$ have been deleted from the original trace,
denoted by $\lambda_{\H}$. 
Similarly a \emph{low projection} of a trace is the sequence of labels in $\L$ that remains after all labels in $\H$ have been deleted from the original,
denoted by $\lambda_{\L}$. 

\begin{definition}[Trace Consistency]
Let $t_1, t_2 \in \Sigma^*$ be \emph{consistent} with each other, written $t_1 \approx t_2$ iff $\exists t \in trace(\LLL)$ such that $t_1$ and $t_2$ are both projections (can be high or low) of $t$.
\end{definition}


Intuitively, non-interference requires that any observable low trace projection produced by the LTS must be consistent with every high trace projection that can be produced by the LTS.
\begin{definition}[Non-interference]
\label{def:ni}
Let $\LLL$ be an LTS, $\Sigma=(\H,\L)$. Let $trace_{\L}(\LLL) = \{l \in {\L}^* \mid \exists t \in trace(\LLL) \wedge t_{\L} \approx l\}$ and similarly define $trace_{\H}(\LLL)$. We say $\LLL$ satisfies \emph{non-interference} iff 
\[ \forall l \in trace_{\L}(\LLL), \forall h \in trace_{\H}(\LLL): l \approx h.\]
\end{definition}

This is a quite restrictive property but intuitive to derive from the original definition of non-interference. 

Non-interference can be expressed in terms of opacity.  Suppose $\LLL$ is an LTS that satisfies non-interference. Assuming that there is more than one high trace in $trace_H(\LLL)$ implies that every high trace projection produced by the LTS is opaque since for a given trace, $t \in trace(\LLL)$, with high trace projection $h$ and low trace projection $l$, every possible high trace projection is consistent with $l$ so there is a trace $t'$ with the same low projection but a different high projection which \emph{covers} $t$.  
Formally, any non-interference property in the sense of Definition~\ref{def:ni} can be reduced to an opacity property with a corresponding static observation function $\obs$.
Let $\MMM=(\LLL,\eta,\obs)$, 
where $\LLL=({\H \cup \L}, S, \FS, F)$ be an LTS and satisfy non-interference.
Consider $\psi$ be a path formula over $\RUN{\LLL, s}$ so that,
$\forall \pi \in \RUN{\LLL, s}.\MMM\models_{\pi} \psi$ iff
for any $\pi' \in \RUN{\LLL, s}$: $\TR{\pi}_{\L} \neq \TR{\pi'}_{\L} \lor \TR{\pi}_{\H} = \TR{\pi'}_{\H}$.
To transform the non-interference property to an opacity property,
we propose observation function $\obs$ such that only low-level labels are observable, \ie
$\forall \pi \in \RUN{\LLL, s}.\obs(\pi) = \TR{\pi}_{\L}$.
$\LLL$ satisfying non-interference implies that 
for any path $\pi$ with high projection $\TR{\pi}_{\H}$
and low projection $\TR{\pi}_{\L}$, 
every possible high trace projection is consistent with $\TR{\pi}_{\L}$,
so there is $\pi'\in \RUN{\LLL, s}$: $\TR{\pi'}_{\L} = \TR{\pi}_{\L} \land \TR{\pi'}_{\H} \neq \TR{\pi}_{\H}$,
\ie $\exists \pi' \in [\![\neg \psi]\!].(\obs(\pi') = \obs(\pi))$,
this implies that $\psi$ is opaque with respect to $\obs$,
and thus can be expressed as $\OPAC{\psi}$.

On the other hand, suppose that $\LLL$ is an LTS in which every high trace projection is opaque. This is not sufficient to imply non-interference as the set of properties is not strong enough. In particular an LTS whose trace set is $\{h_1l_1, h_2l_1, h_2l_2, h_3l_1, h_3l_2\}$ satisfies opacity but not non-interference. Schoepe and Sabelfeld \cite{SchoepeS15} prove equivalence between the two notions for input-output  (i.e. length two) traces when the set of opaque properties is strong enough to characterise every possible information leak. They also have examples that demonstrate the utility and flexibility of opacity in comparison to non-interference.

%% file: opac-pctl.tex
\section{Probabilistic Opaque Temporal Logic \opacpctl}
\label{sec:opacpctl}
This section extends \opacctl with probabilistic operators that allow the construction of probabilistic models
for  quantitative opacity analysis and verification. 
For the purpose of security analysis,
the notion of \emph{probabilistic opacity}~\cite{BryansKM12} defines 
the likelihood of a path in predicate $\varphi$
which is not covered by a path in $\bar\varphi$ from the observer's view:
the smaller the notion the more secure the system. 
\begin{definition}[Degree of opacity]
The degree of $\psi$-opaque property is defined as:
\[
{\D}(\OPAC{\psi}) = 
{\prob}(\psi \setminus \obs^{-1}(\obs(\bar \psi))),
\]
\ie the probability that a $\psi$-trace is not opaque,
where ${\prob}(x)$ denotes the probability of $x$,
$\psi$ and $\bar \psi$ represent the set of traces satisfying and violating property $\psi$ respectively.
\end{definition}
We propose to add a probabilistic operator in our logic to capture this definition (the degree of opacity).

\subsection{Probabilistic model}

We generalise Definition~\ref{def_lts} to cater for
probabilities. Below $\dist(X)$ denotes the set of discrete
probability distribution over a set $X$.

\begin{definition}
A \emph{probabilistic labelled transition system} (pLTS) is a tuple
$\LLL = (S, \Sigma, \PPP, F)$. Here $S$ and $\Sigma$ are states, and
labels, respectively.  $\PPP : S \to {\dist}(\Sigma \times S)$ is the
probabilistic transition relation.
\end{definition}

\begin{definition}
Given a pLTS $\LLL = (S, \Sigma, \PPP, F)$, we construct an LTS $(S,
\Sigma, \FS, F)$ by setting $\FS\ = \{(s, l, s')\ |\ \PPP(s)(l,s') >
0\}$. By ``abus de notation'' we will also use $\LLL$ to refer to
the induced LTS $(S, \Sigma, \FS, F)$.  We say a pLTS is
\emph{deterministic} if the induced LTS is deterministic in the sense
of Section~\ref{sec:prelim}.
\end{definition}

\begin{definition}
A \emph{probabilistic model} is a tuple $\MMM = (\LLL, \eta, obs)$ where:
$\LLL$ is a deterministic pLTS, $\eta$ is a state labelling function,
and $obs$ is an observation function.
Each probabilistic model $\MMM$ induces a model in the sense of \S
\ref{def_model}, by replacing $\MMM$'s pLTS with the induced
LTS. Abusing notation once more, we will also use $\MMM$ to denote
this induced model.
\end{definition}

\subsection{\opacpctl}
To allow quantitative verification, we add the probabilistic operators to the state formulae:
\[
      \phi \ ::=\  ... 
      ~|~ \P_{\bowtie p} \lbrack \psi \rbrack
      ~|~ \P_{\bowtie p} \lbrack \OPAC{\psi} \rbrack
            \qquad            \qquad
      \psi \ ::= \ ...
\]
Here ${\bowtie} \in \{\le, <, \ge, >\}$, $p \in \lbrack 0,1 \rbrack$.  
A state satisfies a probabilistic operator $\P_{\bowtie p} \lbrack  \OPAC{\psi}
\rbrack$ if the quantity of $\OPAC{\psi}$ is ${\bowtie} p$. 
It is standard to extend $\P_{\bowtie p}$ to path formulae $\psi$ as PCTL, \ie $\P_{\bowtie p} \lbrack \psi \rbrack$,
the procedure can be found in \eg~\cite{Baier2008}. 
The semantics of the probabilistic opacity operator is given as:
\begin{eqnarray*}
  \MMM \models_s \P_{\bowtie p} \lbrack \OPAC{\psi} \rbrack  
      & \qquad\text{iff}\qquad &  
      \D (\OPAC{\psi}) \bowtie p
\end{eqnarray*}

\noindent Note that the satisfaction relations $\models_{\pi}$ and
$\models_{s}$ work on the induced model in the sense of \S
\ref{def_model}, not the probabilistic model itself. This is standard
in probabilistic model checking, see e.g.~\cite{Baier2008}.


\begin{theorem}
Given a probabilistic model $\MMM = (\LLL, \eta, obs)$ and a state $s$
in $\LLL$. Let $\psi$ be a path formula, and: 
$\Pi = \{\pi \in \RUN{\LLL, s}\ |\ \MMM\models_{\pi} \psi  \land \pi~\text{is semantically observable} \}$,
then we have:
   \begin{eqnarray*}
   \MMM \models_s \P_{\bowtie p} \lbrack \OPAC{\psi} \rbrack
      & \quad \text{iff}  \quad &
   \prob ( \Pi )  \bowtie p 
   \end{eqnarray*}
   Here semantically observable is to be understood in the sense of
   Definition \ref{defn_opacity}.
\end{theorem}
\begin{proof}
By the definition of the semantics of the probabilistic operator, we have: 
\begin{eqnarray*}
\MMM \models_s \P_{\bowtie p} \lbrack \OPAC{\psi} \rbrack 
& \Leftrightarrow &
\D(\OPAC{\psi}) \bowtie p 
~ \Leftrightarrow ~
\prob(\psi \setminus \obs^{-1}(\obs(\bar \psi))) \bowtie p \\
& \Leftrightarrow &
\prob(\{\pi \in \RUN{\LLL,s} \mid \MMM \models_{\pi} \psi\} 
~ \setminus ~\\
&& 
\quad \{\obs^{-1}(\obs(\pi')) \in \RUN{\LLL,s} \mid \MMM \models_{\pi'} \neg\psi \}
) ~ \bowtie p ~\\
& \Leftrightarrow &
\prob(\{\pi \in \RUN{\LLL,s} \mid \MMM \models_{\pi} \psi ~ \land ~ \\
&& 
\quad \not\exists \pi'.(\MMM \models_{\pi'} \neg\psi \land \obs(\pi) = \obs(\pi'))\} 
) \bowtie p \\
& \Leftrightarrow &
\prob(\{\pi \in \RUN{\LLL,s} \mid \MMM \models_{\pi} \psi ~\land~ \\
&& 
\quad \pi \text{ is not covered by a path violating } \psi\}) \bowtie p \\
& \Leftrightarrow &
\prob(\Pi) \bowtie p 
\qquad \qquad \qquad \qquad \qquad \qquad \qquad \qquad \hfill\Box
\end{eqnarray*}
\end{proof}

\input{eg-opacpctl}

\input{verification}

%% file: eg-opacpctl.tex
\begin{example}
\label{eg:opacpctl}
Consider the model presented in Fig.~\ref{fig:eg-opacpctl}.
Let $s_3$ be a sensitive state, $s_0$ be a starting state, 
$\{s_3,s_6\}$ be final states of interests,
and the observation function be: 
$a \to a,\, b\to b,\, c\to \epsilon$.
Consider the security property be \emph{eventually reaching $s_3$}, \ie $\psi= \F s_3$, so:
\[
tr(\psi) = \{ ac(b)^*a(\bot)\!^{*}\}, \qquad
tr(\neg \psi) = \{aba(c)^*(\bot)\!^{*}\}.
\]
Let $0.1$ be the threshold quantity, 
$\obs(\psi) = \{a(b)^*a(\bot)\!^{*}\}$,  
$\obs(\neg \psi) = \{aba(\bot)\!^{*}\}$, so:
\begin{eqnarray*}
\P_{\le 0.1} ( \OPAC{\psi})
& \Leftrightarrow& \P_{\le 0.1} \{\psi \setminus \obs^{-1} (\obs(\neg\psi) )\}\\
& \Leftrightarrow& \P_{\le 0.1} (\{aca\bot\!^{*}, acbba\bot\!^{*}, \dots\}) \\
& \Leftrightarrow& \frac{1}{3} * \frac{1}{2} + \frac{1}{3} *(\frac{1}{2})^2*\frac{1}{2} + \dots + \frac{1}{3} *(\frac{1}{2})^n \frac{1}{2} \le 0.1 \\
& \Leftrightarrow& \frac{1}{4} \le 0.1 = \false
\end{eqnarray*}
\begin{figure}[!h]
\centering
  \begin{tikzpicture}[->,>=stealth',shorten >=1pt,auto,node distance=1.5cm, scale = 1, transform shape]
  \node[initial,state] (A) {$s_0$};
  \node[state] (B) [right of=A] {$s_1$};
  \node[state] (C) [above right of=B] {$s_2$};
  \node[state, accepting,fill={rgb:black,1;white,2}] (D) [right of=C] {$s_3$};
  \node[state] (E) [below right of=B] {$s_4$};
  \node[state] (F) [right of=E] {$s_5$};
  \node[state, accepting] (G) [right of=F] {$s_6$};

   \path[every node/.style={font=\sffamily\small,
   		fill=white,inner sep=1pt}]
     (A) edge node [above=1mm]{$1.a$} (B)
     (B) edge [bend left=20] node[left=1.5mm] {$\frac{1}{3}.c$} (C)
	 				 edge [bend right=20] node[left=1.5mm] {$\frac{2}{3}.b$} (E)
     (C) edge [loop above] node[above=1mm] {$\frac{1}{2}.b$} (C)
     			 edge node [above=1mm]{$\frac{1}{2}.a$} (D)
     (D) edge [loop above] node[above=1mm] {$\bot$} (D)
     (E) edge node [below=1mm]{$1.a$} (F)
     (F) edge [loop above] node[above=1mm] {$\frac{1}{2}.c$} (F)
     			 edge node [below=1mm]{$\frac{1}{2}.\bot$} (G)
     (G) edge [loop above] node[above=1mm] {$\bot$} (G)
;
\end{tikzpicture}
\caption{Example: \opacpctl}
\label{fig:eg-opacpctl}
\end{figure}
\end{example}

%% file: verification.tex
\subsection{Verification of \opacpctl}
%
%
%
%
%
Intuitively, verification of probabilistic opacity answers the question 
``is the system opaque?'' quantitatively, relative to a secret property $\psi$ 
and the observability of the adversary given a threshold probability $p$.
Given a probabilistic model $\MMM$, a starting state $s$, 
and a property $\psi$ required to be secure, 
the probabilistic verification problem of opacity property $\OPAC{\psi}$ 
is to decide whether 
$\MMM \models_s \P_{\bowtie p} (\OPAC{\psi})$ holds or not.
%
Algorithm~\ref{algo:p-opac} presents the procedure of finding all \textit{probabilistic non-opaque (observable)} traces 
$p\Lambda(s,\bar \odot\psi)$ starting at $s$. 
Note that the probability of each formatted traces satisfying $\psi$ is also calculated and associated with the trace for quantitative purpose.
Then we can calculate: 
\[
{\D}(\OPAC{\psi}) = \sum_{p\lambda \in p\Lambda (s,\bar \odot\psi)} \prob(p\lambda). 
\]
\input{algo-opacP}

\begin{theorem}[Soundness of $\P_{\MMM} (\OPAC{\psi})$ translation]
Given a model $\MMM$, starting state $s$, a probability threshold $p$, and a security property $\psi$:
\[
\MMM \models_s \P_{\bowtie p} (\OPAC{\psi}) \qquad 
\text{iff} \qquad 
{\prob}(\PT_{\opac}(\MMM,s,\psi) ) \bowtie \mathit{p}.
\]
\end{theorem}
\begin{proof}
The proof is obtained by the satisfaction relation of $\P_{\bowtie p} (\OPAC{\psi}) $ and the construction translation of $\PT_{\opac}(\MMM,s,\psi)$ described in Algorithm~\ref{algo:p-opac}. 
The algorithm will terminate since $\SAT(\psi)$ are computed as a set of regular-expression-like formatted traces satisfying $\psi$ as Algorithm~\ref{algo:opac}. Probability of such a trace is calculated by multiplication of the probability of each transition label for non-cycle part, and multiplication of $p/(1-p)$ for a cycle with probability $p$.
$\hfill\Box$
\end{proof}

We have implemented our translation as an extension of the PRISM model checker~\cite{impl}. Example~\ref{eg:veri} presents the result generated by the prototype tool.
\input{eg-veri}

%% file: algo-opacP.tex
\small{\begin{algorithm}[!h]
 \SetAlgoLined
  \KwData{$\MMM, s, \psi$}
  \KwResult{ probabilistic non-opaque traces $p\Lambda(s,\bar{\opac}\psi)$} 
 \Switch{$\psi$} {
	\textbf{case} $\X \phi$: \ \ \ $\SAT(\psi) \leftarrow \cup_{i \in \NAT}\{tr(s \to s_i) \mid \POST(s)=s_i  \land s_i \in \SAT(\phi)\}$, \\ 
	$\qquad \qquad \quad \ \SAT(\neg\psi) \leftarrow \cup_{i \in \NAT}\{tr(s \to s_i) \mid \POST(s)=s_i  \land s_i \in \SAT(\neg\phi)\}$\;
	\textbf{case} $\phi \U \phi'$: $\SAT(\psi) \leftarrow \textbf{compU}(\MMM, s, \phi, \phi')$, \\ 
	$\qquad \qquad \quad \ \SAT(\neg\psi) \leftarrow \textbf{compR}(\MMM, s, \neg\phi, \neg\phi')$\;
	\textbf{case} $\phi \R \phi'$: $\SAT(\psi) \leftarrow \textbf{compR}(\MMM, s, \phi, \phi')$, \\ 
	$\qquad \qquad \quad \ \SAT(\neg\psi) \leftarrow \textbf{compU}(\MMM, s, \neg\phi, \neg\phi')$\;
	}  
	$p\Lambda \leftarrow  \{p\lambda \mid p\lambda.tr \leftarrow \lambda \land p\lambda.pr \leftarrow \prob(\lambda) \ \text{for} \ \lambda \in \SAT(\psi)\}$\;
	$p\Lambda' \leftarrow \{p\lambda \mid p\lambda.tr \leftarrow \lambda \land p\lambda.pr \leftarrow \prob(\lambda) \ \text{for} \ \lambda  \in \SAT(\neg\psi)\}$\;
   $p\Lambda'' = \{\}$\;
  \For{each $p\lambda \in p\Lambda$}
  {
		\For {each $p\lambda' \in p\Lambda'$}
		{
	  	\If{$\obs(p\lambda.tr) \subseteq \obs(p\lambda'.tr)$}
	  	{ 
	    	$p\Lambda'' \leftarrow p\Lambda'' \cup \{p\lambda\}$; break \;
	  	}
	  }
  }
  ${p\Lambda}_{\bar{\opac}} \leftarrow p\Lambda \setminus p\Lambda''$; \tcc*[f]{non-opaque probabilistic traces}\;
  \Return{$p\Lambda_{\bar{\opac}}$}.
  \caption{$\PT_{\opac}(\MMM,s,\psi)$: finding probabilistic non-opaque traces} 
 \label{algo:p-opac}
\end{algorithm}}

%% file: eg-veri.tex
\begin{example}
\label{eg:veri}
Consider the following example used in~\cite{BryansKM12} presented in Fig.~\ref{fig:eg-probopac}.
Let $s_0$ be a starting state, 
$\{t_0, t_1, t_2, t_3, t_4, t_5\}$ be final states,
observation function be: $a \to a$, $b \to \epsilon$, $c \to c$, and $x \to \epsilon$.
Assume the property of interest is the system \emph{terminating at sensitive states $\{t_2, t_3, t_5\}$}.
%
PRISM also allows to directly specify properties which evaluate to a value using $\P=?[\psi]$.
The property specification is given as:
\[
\P =? [\opac ~\F~ (((s=2) \lor (s=3) \lor (s=5)) \land (t=1))],
\]
where $s=i$ denotes $t_i$ is a sensitive state, $t=1$ denotes the status of terminating.
We can automatically calculate the probabilistic opacity of the system as:
\begin{Verbatim}
   Result: 0.026.{0.01046:bcax:ca,0.0156:ca(b)*x:ca} 
           (value in the initial state)
\end{Verbatim}
where $0.026$ denotes the probability of opacity of the system, 
\ie the probability of traces satisfying $\psi$ 
but not covered by those observationally equivalent traces violating $\psi$, which include: 
$bcax$ with probability $0.0104$ and $ca(b)^*x$ with probability $0.0156$, both observed as $ca$.
\begin{center}
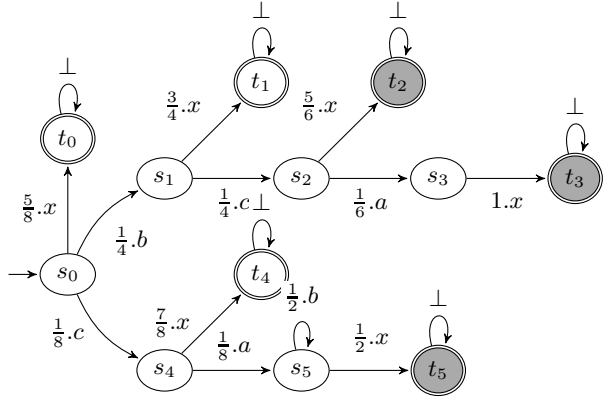
\begin{figure}[!h]
\centering
  \begin{tikzpicture}[->,>=stealth',shorten >=1pt,auto,node distance=1.8cm, scale = 1, transform shape]

  \node[initial, state] (A) {$s_0$};
  \node[state, accepting] (A1) [above of=A]  {$t_0$};
  \node[state] (B) [above right of=A] {$s_1$};
  \node[state, accepting] (B1) [above right of=B] {$t_1$};
  \node[state] (C) [right of=B] {$s_2$};
  \node[state, accepting,fill={rgb:black,1;white,2}] (C1) [above right of=C] {$t_2$};
  \node[state] (D) [right of=C] {$s_3$};
  \node[state, accepting,fill={rgb:black,1;white,2}] (D1) [right of=D] {$t_3$};
  \node[state] (E) [below right of=A] {$s_4$};
  \node[state, accepting] (E1) [above right of=E] {$t_4$};
  \node[state] (F) [right of=E] {$s_5$};
  \node[state, accepting,fill={rgb:black,1;white,2}] (F1) [right of=F] {$t_5$};

   \path[every node/.style={font=\sffamily\small,
   		fill=white,inner sep=1pt}]
     (A) edge node [left=1mm]{$\frac{5}{8}.x$} (A1)
	 edge [bend left=20] node[below right=1.5mm] {$\frac{1}{4}.b$} (B)
	 edge [bend right=20] node[left=2mm] {$\frac{1}{8}.c$} (E)
     (B) edge node [above left=1.5mm]{$\frac{3}{4}.x$} (B1)
	 edge node[below=1mm] {$\frac{1}{4}.c$} (C)
     (A1) edge [loop above] node[above=1mm] {$\bot$} (A1)
     (B1) edge [loop above] node[above=1mm] {$\bot$} (B1)
     (C1) edge [loop above] node[above=1mm] {$\bot$} (C1)
     (D1) edge [loop above] node[above=1mm] {$\bot$} (D1)
     (E1) edge [loop above] node[above=1mm] {$\bot$} (E1)
     (F1) edge [loop above] node[above=1mm] {$\bot$} (F1)
     (D) edge node[below=1.5mm] {$1.x$} (D1)
     (C) edge node[above left=1.5mm] {$\frac{5}{6}.x$} (C1)
     	 edge node [below=1mm]{$\frac{1}{6}.a$} (D)
     (E) edge node[left=2.5mm] {$\frac{7}{8}.x$} (E1)
     	 edge node [above=1mm]{$\frac{1}{8}.a$} (F)
     (F) edge [loop above] node[above=1mm] {$\frac{1}{2}.b$} (F)
     	 edge node [above=2mm]{$\frac{1}{2}.x$} (F1)
;
\end{tikzpicture}
\caption{Example: probabilistic opacity, grey node denotes a sensitive state.}
\label{fig:eg-probopac}
\end{figure}
\end{center}
\end{example}

%% file: case-study.tex
\section{Implementation and Examples}
\label{sec:eg}

We have built a prototype tool for verification of opacity as an extension of the PRISM model checker~\cite{KNP11}. 
Models are described in an extension of the PRISM modelling language with observations and transition labels.
The new model type is denoted as ``ldtmc''. Properties are described in an extension of the PRISM's property specification language with the opacity operator.
The tool and details of all examples and case studies are available from~\cite{impl}.

\subsection{Modelling probabilistic opacity in PRISM}
Models in PRISM are described in a state-based language based on guarded commands. 
A model is constructed as a number of \emph{modules} which can interact.
Each module contains a set of finite-valued variables which define the state of the module. 
The behaviour of each module is described by a set of guarded commands in the form of: 
\begin{small}
\begin{Verbatim}
   [<action>] <guard> -> <prob> : <update> + ... + <prob> : <update>;
\end{Verbatim}
\end{small}
The \emph{guard} is a predicate over the variables of all the modules in the model.
Each \emph{update} describes a probabilistic transition 
which specifies how the variables of the module are updated if the guard is true.
The \emph{prob} attached with each update specifies the probability that the corresponding state transition takes place.
The \emph{action} label is optional which allows modules to synchronise over commands.

We have extended the existing modelling language for model type ``ldtmc'' to allow: 
(i) the definition of \emph{observation functions}: 
\begin{small}
\begin{Verbatim}
	observations
		<label> -> <observable>, ... <label> -> <observable>;
	endobservations
\end{Verbatim}
\end{small}
which defines each label and its observation 
through the keyword \texttt{observations};
(ii) and the specification of \emph{transition labels} over updates is in the form:
\begin{small}
\begin{Verbatim}
   []<guard> -> <prob>:<LABEL>:<update> + ... + <prob>:<LABEL>:<update>;
\end{Verbatim}
\end{small}

\subsection{Modelling dining cryptographers}
 Anonymity is an important concept in security.  To illustrate the
 versatility of our work, we use our logic to express anonymity of the
 dining cryptographers protocol~\cite{Chaum88}.
 Consider that three cryptographers $1$, $2$ and $3$ are sharing a meal at a restaurant.
 At the end of the meal, 
 at most one cryptographer will pay the bill,
 and they would like to check whether the bill has
 been paid or not, but the cryptographers respect each other's right
 to make an anonymous payment.
 A two-stage protocol is performed to solve the problem:
 (i) every two cryptographers establish a shared one-bit secret:
     each of them flips a coin, the outcome is only visible to
     himself and the cryptographer on his right;
 (ii) each cryptographer publicly announces whether the two outcomes agree or disagree,
      if the cryptographer is not the payer, he says the truth,
      otherwise, he states the opposite of what he sees.
 When all cryptographers have announced, they count the number of disagrees. 
 If that number is odd, then one of them has paid,
but no other cryptographer is able to deduce who is the payer. 

Without loss of generality, let us assume cryptographer 3 is the observer who tries to know which of the other two paid the bill
if the bill has been paid.
The observer does not know the initial state of 
the cryptographer $i=1,2$ being the payer ($p_i$), 
he knows the outcome of the flipping coins of himself and of the cryptographer-1: \emph{head} ($h_i$) or \emph{tail} ($t_i$) where $i=1,3$ but does not know that of the cryptographer-2,
 he knows the procedure of the protocol and he can hear 
 what cryptographer $i$ says: \emph{agree} ($a_i$) or \emph{disagree} ($d_i$),
 and therefore he knows the outcome of the disagreement counting is 
 \emph{even} ($e$) or \emph{odd} ($o$).
 In other words, labels $p_i$, $h_2$, $t_2$ are invisible to him,
 while labels $h_1$, $t_1$, $h_3$, $t_3$, $d_i$, $a_i$, $e$ and $o$ 
 are visible to him.
 Therefore the set of observables includes:
 $\Theta = \{d_i, a_i, e, o ~|~ 1 \le i \le 3\} \cup \{h_i, t_i ~|~ i=1,3\}$,
 and the observation function on the transition labels of the model is specified as:
 $p_1,p_2,t_2,h_2 \to \epsilon$;
 $a_i \to a_i$, $d_i \to d_i$; $h_1 \to h_1$;
 $t_1 \to t_1$, $h_3 \to h_3$; $t_3 \to t_3$;
 $e \to e$; $o \to o$.
Our tool can automatically check the property ``cryptographer $1$ is the payer'' ($\psi = \X (\mathit{payer}=c_1)$) is \emph{opaque} if the bill has been paid.
The property specification is given as:
$\P=? [\opac~ \X~ (\mathit{payer}=c_1)]$.
Our  tool answers  the question  ``\emph{$c_1$ is the payer} is $\psi$-opaque'' as follows:
\begin{Verbatim}
   Result: 0.0.{} (value in the initial state)
\end{Verbatim}
The result shows that there are no observable traces found.

\subsection{Modelling a location privacy example}
Consider a simple example of location privacy releasing by credit card
records presented in Fig.~\ref{fig:case-location} which describes the card holder's activities.  Assume the
adversary can observe the credit card records to track partial
location information and the observations are given as:
$\mathsf{station} \rightarrow s, \mathsf{work} \rightarrow \epsilon, \mathsf{travel} \rightarrow
\epsilon, \mathsf{office} \rightarrow \epsilon, \mathsf{coffeshop} \rightarrow c, \mathsf{bankA}
\rightarrow b, \mathsf{bankB} \rightarrow b, \mathsf{airport} \rightarrow a, \mathsf{home}
\rightarrow \epsilon, \mathsf{L1} \rightarrow \epsilon, \mathsf{L2} \rightarrow
\epsilon, \mathsf{L3} \rightarrow \epsilon$.
Suppose that the states leading to the final location $\mathsf{L1}$, $\mathsf{L2}$ and $\mathsf{L3}$ are sensitive. 
Then we have the property specification as:
$\P=? [\opac~ \F~ dest]$,
where \emph{dest} denotes states $q_9$, $q_{10}$ and $q_{11}$, led by locations $\mathsf{L1}$, $\mathsf{L2}$ and $\mathsf{L3}$ respectively.
The result generated by the tool is as follows:
\begin{Verbatim}
   Result: 0.33333333333.
           {0.16666666666666666:stationtravelbankBairportL1:sba,
           0.16666666666666666:stationtravelbankBairportL2:sba} 
           (value in the initial state)
\end{Verbatim}
The result meets our intuition, that the traces leading to $\mathsf{L1}$ and $\mathsf{L2}$
with observation $sba$ and $sba$ are not covered by traces leading to insensitive final location states.
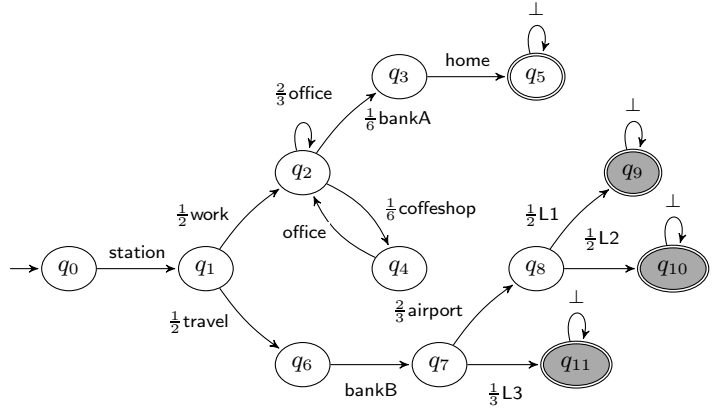
\begin{figure}[!h]
\centering 
  \begin{tikzpicture}[->,>=stealth',shorten >=1pt,auto,node distance=1.8cm, scale = 1, transform shape]
  \node[initial,state] (A) {$q_0$};
  \node[state] (B) [right of=A] {$q_1$};
  \node[state] (C) [above right of=B] {$q_2$};
  \node[state] (D) [above right of=C] {$q_3$};
  \node[state] (E) [below right of=C] {$q_4$};
  \node[state, accepting] (F) [right of=D] {$q_5$};
  \node[state] (G) [below right of=B] {$q_6$};
  \node[state] (I) [right of=G] {$q_7$};
  \node[state] (J) [above right of=I] {$q_8$};
  \node[state, accepting,fill={rgb:black,1;white,2}] (K) [above right of=J] {$q_9$};
  \node[state, accepting,fill={rgb:black,1;white,2}] (L) [right of=J] {$q_{10}$};
  \node[state, accepting,fill={rgb:black,1;white,2}] (M) [right of=I] {$q_{11}$};

   \path[every node/.style={font=\sffamily\scriptsize,
   		fill=white,inner sep=1pt}]
     (A) edge node [above=1mm]{station} (B)
     (B) edge [bend left=10] node[left=2mm] {$\frac{1}{2}$work} (C)
	 edge [bend right=10] node[left=2mm] {$\frac{1}{2}$travel} (G)
     (C) edge [loop above] node[above=2mm] {$\frac{2}{3}$office} (C)
     	 edge [bend left=10] node [right=2mm]{$\frac{1}{6}$bankA} (D)
     	 edge [bend left=20] node [right=2mm]{$\frac{1}{6}$coffeshop} (E)
     (D) edge node[above=1mm] {home} (F)
     (E) edge [bend left=20] node [left=1mm]{office} (C)
     (F) edge [loop above] node[above=1mm] {$\bot$} (F)
     (G) edge node[below=2mm] {bankB} (I)
     (I) edge [bend left=10] node[left=2mm] {$\frac{2}{3}$airport} (J)
	 edge node[below=2mm] {$\frac{1}{3}$L3} (M)
     (J) edge [bend left=10] node[left=2mm] {$\frac{1}{2}$L1} (K)
	 edge node[above=2mm] {$\frac{1}{2}$L2} (L)
     (K) edge [loop above] node[above=1mm] {$\bot$} (K)
     (L) edge [loop above] node[above=1mm] {$\bot$} (L)
     (M) edge [loop above] node[above=1mm] {$\bot$} (M)
;
\end{tikzpicture}
\caption{Modelling a location privacy example}
\label{fig:case-location}
\end{figure}


%% file: entropy.tex
\section{Entropy of the Opacity Formula}
\label{sec:entropy}
Probabilistic specification $\P(\OPAC{\psi})$ calculates the probability of non-$\psi$-opaque behaviours of a model $\MMM$.
In this section, we provide some discussions on an alternative measurement of the opacity formula based on the notion of \emph{entropy}:
\[\HHH(\OPAC{\psi}).\]
We adapt the definition of the entropy of a language (of finite words)~\cite{ChomskyM58} to calculate this.

\begin{definition}[$\HHH \lbrack \OPAC{\psi} \rbrack$]
\label{def:entropy}
Given a model $\MMM = (\LLL, \eta, \obs)$. 
Let $\psi$ be a path formula,
the entropy of the opacity formula $\OPAC{\psi}$ is defined as:
\[
\HHH ( \OPAC{\psi}  ) 
= \limsup_{n \to +\infty} \frac{\log_2 (1+|~\lbrack \psi \setminus \obs^{-1}(\obs(\bar\psi))\rbrack_n~|)}{n},
\]
where:
\begin{eqnarray*}
|~\lbrack \psi \setminus \obs^{-1}(\obs(\bar\psi))\rbrack_n~|
&=&
|~\{
\pi \in \RUN{\LLL,s} \mid \MMM \models_{\pi} \psi\\ 
&& 
\quad \land ~ 
|~\ERASE{\pi}~|=n\\ 
&& 
\quad \land ~ 
\pi \text{ is transparent}
\}~|.
\end{eqnarray*}
\end{definition}
Intuitively, the entropy of $\OPAC{\psi}$ can be understood as the amount of information (in bits per symbol) in typical words of (the language of) ${\psi}$-transparent.  

It is easy to slightly change Algorithm~\ref{algo:opac} to calculate the size of transparent traces. We can then take the prefix of those traces with length of $n$ and compute the entropy of $\psi$-opaque property using Definition~\ref{def:entropy}.

\begin{example}
Consider again the models presented in Example~\ref{fig:eg-opacctl}. It is easy to see that for any $n \in \NAT$, 
$|~ \lbrack \OPAC{F s_i} \rbrack_n ~| = 0$, so:
\[
\HHH \lbrack \OPAC{F s_i} \rbrack
= \limsup_{n \to +\infty} \frac{\log_2 (1+0)}{n}
= 0,
\]
where $i=3$ for model (a) and $i=2$ for model (b).
The result shows the degree of transparency of traces satisfying $\psi = F s_i$ is $0$ in the notion of entropy, and implies the model is $\lbrack F s_i \rbrack$-opaque.
\end{example}

\begin{example}
Consider again the models presented in Fig.~\ref{fig:eg-opacpctl}. It is easy to see that for any $n \in \NAT$, 
$|~ \lbrack \OPAC{F s_3} \rbrack_n| = n$, so:
\[
\HHH \lbrack \OPAC{F s_3} \rbrack
= \limsup_{n \to +\infty} \frac{\log_2 (1+n)}{n}
= 0.
\]
The result shows the degree of transparency of traces satisfying $\psi = F s_3$ is $0$ in the notion of entropy, and implies the model is $\lbrack F s_3 \rbrack$-opaque.
Note that, 
the entropy-based measurement is not as precise as the probabilistic-based measurement.
But it somehow meets our intuition:
when $n \to \infty$,  
most of the traces satisfying $\lbrack F s_3 \rbrack$ are covered by traces violating $\lbrack F s_3 \rbrack$ (only $acba\bot^*$ is not covered) from the observer's view.
\end{example}

\cite{AsarinBDDM14a} formulated the basic entropy-based properties in model checking context. We adapt their discussions here to our scenario for opacity properties, to illustrate some intuition of entropy-based measurement.
Consider a model $\MMM$, and an opacity formula $\OPAC{\psi}$. Let $\Pi$ and $\Pi(\OPAC{\psi})$ denote all the behaviours of the model and all the (in)finite paths satisfying $\OPAC{\psi}$, intuitively: 
\begin{itemize}
\item $\HHH(\Pi)$ measures the quantity of all the behaviours of the system, 
\item $\HHH(\Pi \cap \Pi(\OPAC{\psi}))$ measures the quantity of the $\psi$-opaque behaviours of the system, \item $\HHH(\Pi)-\HHH(\Pi \cap \Pi(\OPAC{\psi}))$ quantifies how difficult to steer the system to be $\psi$-opaque, 
\item and $\HHH(\Pi \setminus  \Pi(\OPAC{\psi})) = \HHH(\OPAC{\psi})$ measures the quantity of non-$\psi$-opaque (transparent) behaviours of the system.
\end{itemize}

In general, for any language accepted by a given finite well-structured model $\MMM$, its entropy can be effectively computed using linear algebra~\cite{ChomskyM58}. 
Let $A(\MMM)$ denote the extended adjacency  matrix of $\MMM$:
\[
A(\MMM) = | \{a \in \Sigma \mid s \TRANS{a} t \in \FS\}|.
\]
\begin{theorem}\cite{ChomskyM58}
\label{thm:entropy}
For any finite deterministic trimmed model $\MMM$, 
the entropy of the paths accepted by $\MMM$ can be calculated as:
\[
\HHH(\RUN{\MMM}) = \log_2 \rho(A(\MMM)),
\]
where $\rho(A)$ is the spectral radius the matrix $A$, \ie maximal modulus of its eigenvalues.
\end{theorem}

If we can find a finite deterministic model accepting the specified paths satisfying the property (this is out range of this paper), we can then calculate the entropy of the property as the logarithm of a spectral radius using Theorem~\ref{thm:entropy}. 

\begin{example}
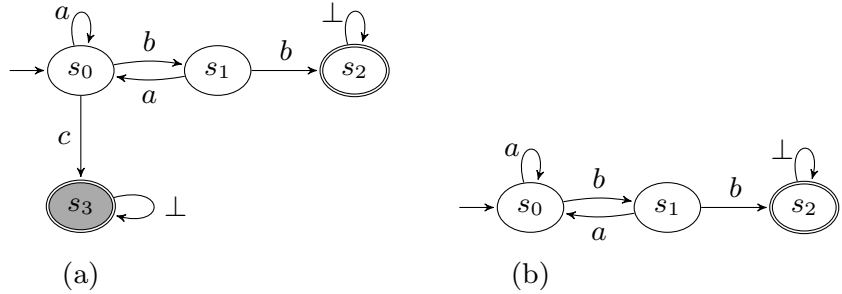
\begin{figure}[!h]
\centering
 \begin{tikzpicture}[->,>=stealth',shorten >=1pt,auto,node distance=1.5cm, scale = 1.2, transform shape]
  \node[initial,state] (0) {$s_0$};
  \node[state] [right of=0] (1) {$s_1$};  
  \node[state, accepting] [right of=1](2) {$s_2$};
  \node[state, accepting,fill={rgb:black,1;white,2}] [below of=0] (3) {$s_3$};
  \node[align = center, below =2mm of 3] {(a)};
  \path (0) edge [loop above] node[left] {$a$} ();
  \path (0) edge  node[left] {$c$} (3);
  \path (0) edge [bend left=10] node[above] {$b$} (1);
  \path (1) edge [bend left=10] node[below] {$a$} (0);
  \path (1) edge  node[above] {$b$} (2);
  \path (2) edge [loop above] node[left] {$\bot$} ();
  \path (3) edge [loop right] node[right] {$\bot$} ();
\end{tikzpicture}
~~~~
 \begin{tikzpicture}[->,>=stealth',shorten >=1pt,auto,node distance=1.5cm, scale = 1.2, transform shape]
  \node[initial,state] (0) {$s_0$};
  \node[state] [right of=0] (1) {$s_1$};
  \node[state, accepting] [right of=1](2) {$s_2$};
  \node[align = center, below =2mm of 0] {(b)};
  \path (0) edge [loop above] node[left] {$a$} ();
  \path (0) edge [bend left=10] node[above] {$b$} (1);
  \path (1) edge [bend left=10] node[below] {$a$} (0);
  \path (1) edge  node[above] {$b$} (2);
  \path (2) edge [loop above] node[left] {$\bot$} ();
\end{tikzpicture}
\caption{Example: Entropy of $\OPAC{F~ s_3}$ for model (a). Grey node denotes a sensitive state. (b) is the model accepting all ${F~ s_3}$-transparent paths.}
\label{fig:eg-entropy}
\end{figure}
Consider a model presented in Fig.~\ref{fig:eg-entropy} (a) named $\MMM_a$. Let $s_3$ be a sensitive state, and $s_0$ be the starting state, and the observation function be: $a \to \epsilon$, $b \to b$, $c \to c$. Consider the security property \emph{eventually reaching $s_3$}, \ie $\psi = \F s_3$ in our logic. 
It is easy to see that the model presented in Fig.~\ref{fig:eg-entropy} (b), say $\MMM_b$, accepts all ${F~ s_3}$-transparent paths.
We can then calculate $\HHH(\OPAC{F~ s_3})$ by Theorem~\ref{thm:entropy} as:
\[
\HHH(\OPAC{F~ s_3}) = \log_2 \rho(A(\MMM_b) = \log_2 \frac{1+\sqrt{5}}{2}.
\]
\end{example}

%% file: related-work.tex
\section{Related Work}
\label{sec:related-work}
We present related literature from the perspective
of specification and verification of information security policies.
This work relates to the specification of information security policies and (quantitative) verification of opacity properties.

\textit{Security policies.} 
There are a number of information security policies for confidentiality:
non-deducibility~\cite{Sutherland86} is designed to keep attacker observable events 
consistent with possible variations of secret inputs,
but this policy is not able to protect secret outputs and to address covert channels;
non-interference~\cite{GoguenMeseguer82} is one of the most popular flow policies but it is too strong for practical applications;
quantified non-interference~\cite{ClarkHM01} is then introduced to relax the absolute non-interference policy by computing the amount of the interference;
declassification policies~\cite{SabelfeldS05} control the release of information;
(quantified) opacity~\cite{BryansKMR08,BerardMS10,BryansKM12,BerardCS15b} is considered as a
more general property where the sensitive information can be contained in the input, output, and each transition step,
which can lead to the definition of initial, final, and language-based opacity.
Opacity is a promising approach for describing and unifying security properties. 
Recently, quantified opacity~\cite{BerardMS10,BryansKM12,BerardCS15b} has been studied
in terms of probability and information entropy. 

This work focuses on general policies using opacity. 
Opacity has reasonable potential being a good choice for specifying flow security properties for modern communication systems due to the feature of partial observability and uncertainty of the environment.

\textit{Verification of opacity.} 
Opacity was first introduced in the context of cryptographic protocols in~\cite{Boisseau03,Mazare04}.
Later on, Bryans et al investigated opacity properties in systems
modelled as Petri nets~\cite{BryansKR04} and labelled transition systems~\cite{BryansKMR08},
where the secret was specified as predicates over system runs.
Generally speaking, the opacity properties can be classified into two families:
\emph{language-based opacity}~\cite{BadouelBBCD07,Dubreil09,Ben-KalefaL11} 
where the secret predicate is regarding to a subset of system runs;
and \emph{state-based opacity}
~\cite{BryansKMR08,SabooriH07,SabooriH11,WuL13,FalconeM13} 
where the secret predicate is referring to a subset of states.
Intuitively, the system is language-based opaque if for any word $w$
in the secret language $L_S$, there is at least one other word $w'$
in the non-secret language $L_{NS}$ equivalent to $w$ based on the adversary's observations;
the system is considered as state-based opaque if the adversary
is not able to induce whether the initial state (initial-state opacity), 
current state (current-state opacity~\cite{BryansKMR08,SabooriH07}), 
the state a few steps ago (K-step opacity~\cite{SabooriH07}), or 
the initial-final state pair (Initial-and-Finial state opacity~\cite{WuL13}) 
is a secret state or not.
%
Dubreil~\cite{Dubreil09} studied opacity verification of infinite-state systems using \textit{approximate} models.
Kobayashi and Hiraishi~\cite{KobayashiH13} investigated the approach of
verification of opacity for infinite-state discrete event systems modelled by pushdown automaton.
They showed that opacity of pushdown systems is undecidable.
The relationship among variant notions of opacity has been studied~\cite{SabooriH11,CassezDM12,WuL13} and transformation mappings between them described -- enabling efficient use of existing verification approaches.
%

This paper focuses on an easy-expressing logic \opacctl and \opacpctl for \emph{specifying} opacity, and applies probabilistic model checking techniques for automatically \emph{verifying} opacity properties with guarantees. 
Automated verification approach is built on solid foundations and provides the rigorous guarantees needed to give confidence and identify subtle flaws in a security system. 
Quantitative aspects enable designers to effectively weigh and arbitrate between concerns such as security and performance. Quantitative verification therefore turns to be a good fit for the analysis of security property. 
We build the prototype tool as an extension of PRISM~\cite{KNP11} for the quantitative verification perspective,
and will consider to integrate opacity enforcement mechanism into our framework as a future work.

\textit{Logics for security properties.}
Linear temporal logic formulas cannot express properties of sets of execution paths, i.e. hyperproperties are required for specifying information flow policies,
such as non-interference properties.
Computational tree logic formulas cannot express observational determinism
even if path quantifiers are defined, 
hence the need to develop specialised logics
that characterise information flow properties for security policies.
Dimitrova \cite{DimitrovaFKRS12} proposed SecLTL, 
an extension of LTL with a hide operator, 
for specifying path-based integration of information flow policies in reactive systems.
The hide operator specifies requirements such that 
the observable behaviour of a system needs to be independent of the choice of a \emph{secret}.
Clarkson et. al \cite{ClarksonFKMRS14} proposed HyperLTL and HyperCTL$^*$ 
as an extension of propositional linear-time temporal logic (LTL) 
and branching-time temporal logic (CTL$^*$) respectively
for the purpose of specifying hyperproperties.
HyperLTL and HyperCTL$^{*}$ specified information flow policies 
by explicit quantification over multiple traces.
Epistemic logic~\cite{GarySyverson92} is a different approach
to reason about information flow properties from perspective of \emph{knowledge},
rather than \emph{secrets} as other logics for information flow properties.
Specifically, \emph{knowledge operators} are introduced to temporal logics 
to express knowledge-based information flow properties
for imperative programs~\cite{BalliuDG11}.

Comparing with the above works, 
we focus on elementarily expressing and verifying observability properties for security systems.
Logic for Hyperproperties can be more expressive than our logic, 
but we aim to propose easy-expressing specification of opacity and observability property for information security concerns in particular.
The logic \opacctl proposed in this paper 
can be used to specifying the opacity property in a very straightforward way for security systems, 
and the property required to be secret can be defined flexibly regarding users' requirements.
For instance, the user can require that a particular state such as initial, next or final of the system should be kept secret.
In addition, our logic \opacpctl can specify opacity security properties in a quantitative way, allowing us to reason about the degree of satisfaction or violation of the security property of interest. 
Such a degree is measured based on both \textit{probability} and \textit{entropy} of observable behaviours of the model, which is novel and promising.

\textit{Quantitative security properties}
Opacity and related concepts were first studied
and related to information flow properties in a \emph{qualitative}
context in~\cite{BryansKMR08}.
In the probabilistic context, opacity has been studied in
\cite{LakhnechM05,BerardMS10,BryansKM12}.
\cite{LakhnechM05} studied the notion of opacity in the probabilistic
computational world. There opacity was based on the probabilities of
observer's pre-beliefs on the truth of the predicate.
The work in \cite{BerardMS10} presents a quantitative information
leakage analysis in terms of probabilistic opacity.
A number of quantitative opacity notions are introduced in~\cite{BryansKM12}
which can be applied in information flow security analysis.

In this paper, the measurement of opacity has been studied in two perspectives. 
We measure the probability of observable behaviours to which extent the model satisfies a sensitive property, and apply probabilistic model checking technique to automate the approach. In many situations probabilistic verification is highly relevant, but there is an important limitation: for many interesting properties, the probability is either 0 or 1 (too precise) and thus no quantitative sense analysis~\cite{AsarinBDDM14a}. 
We therefore also study the entropy of observable behaviours regarding to a sensitive property of interest. This allows us to measure the amount of information in bits per symbol in typical behaviours.

%% file: conclusion.tex
\section{Conclusions}
\label{sec:conc}

We have proposed a novel, probabilistic logic for simply expressing the opacity of labelled transition system properties and demonstrated how opacity in this context can be checked using an extension of the PRISM model checker. 
We also provide a discussion towards an entropy-based measurement of the opacity formulae. 
Given the flexibility of opacity as a security property there are many possible directions for future research. Promising directions include applying our technique to location privacy protocols and generalising the opacity framework to games and robotic systems modelled as partially observable transition systems in order to provide better decision procedures. 
We also propose to develop approaches towards an entropy-based measurement of the opacity formulae.

%% file: appendix.tex
\section*{Appendix}
The Dining Cryptographers protocol can be considered as a network of 
communicating PLTS shown in Fig.~\ref{fig:dc}. 
Fig.~\ref{fig:dc} (a) models the procedure of flipping coins.
There are $3$ states:
 state $0$ is the initial state,
 state $1$ indicates the outcome of one of the cryptographer 
 (let us say the $3^{rd}$ cryptographer) is ``head'' ($h_3$)
 whatever his left/right neighbour get either head ($h$) or tail ($t$),
 and state $2$ indicates the outcome of the $3^{rd}$ cryptographer is ``tail'' ($t_3$)
 whatever his left/right neighbour get either head ($h$) or tail ($t$).
Fig.~\ref{fig:dc} (b) models the behaviours of 
 the $i^{th}$ cryptographer ($i=1,2,3$). 
 There are $5$ states: state $0$ is the initial state 
 and specifies whether $i$ is the payer $i=p_i$ or not $i \ne p_i$,
 state $1$ indicates the cryptographer see two same sides of the coin 
 (either two head $hh_i$ or two tail $tt_i$) regarding to himself and 
 his left-hand side neighbour's toss-up,
 state $2$ indicates the cryptographer see two different outcome of the toss-up,
 state $3$ indicates the announced outcome of the cryptographer is ``disagreed''
 while state $4$ indicates the announced outcome is ``agreed''.
Fig.~\ref{fig:dc} (c) models the procedure of outcome counting (number of difference).
 There are $7$ states:
 state $0$ is the initial state,
 state $1$ (c.f. $2$) indicates the announced outcome of 
 cryptographer $1$ is ``agree'' (c.f. ``disagree''),
 state $3$ (c.f. $4$) indicates the announced outcome of 
 both cryptographer $1$ and $2$ are ``same'' (c.f. ``different''),
 state $5$ (c.f. $6$) indicates the number of announced outcome of 
 ``disagreement'' is even (c.f. ``odd'').
 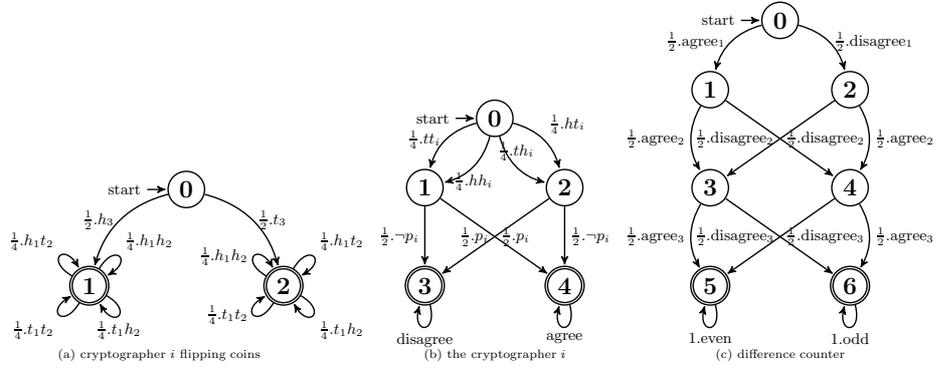
\begin{figure}[h!]
 \begin{center}
 \scalebox{0.65}{
  \begin{tikzpicture}[->,>=stealth',shorten >=1pt,auto,node distance=2cm,
                 thick,main node/.style={circle,draw,font=\Large\bfseries}]
 
   \node[main node, initial] (0) {0};
   \node[main node, accepting] (1) [below left=2cm of 0] {1};
   \node[main node, accepting] (2) [below right=2cm of 0] {2};
   \node[below right of=1] {{\scriptsize(a) cryptographer $i$ flipping coins}};
   
   \path
     (0) edge [bend right] node [left] {$\frac{1}{2}.h_3$} (1)
         edge [bend left] node [right] {$\frac{1}{2}.t_3$} (2)
     (1) edge [loop above right] node {$\frac{1}{4}.h_1h_2$} (1)
 	edge [loop above left] node{$\frac{1}{4}.h_1t_2$} (1)
 	edge [loop below right] node [below]  {$\frac{1}{4}.t_1h_2$} (1)
 	edge [loop below left] node {$\frac{1}{4}.t_1t_2$} (1)
     (2) edge [loop above right] node {$\frac{1}{4}.h_1t_2$} (2)
 	edge [loop above left] node[left] {$\frac{1}{4}.h_1h_2$} (2)
 	edge [loop below right] node {$\frac{1}{4}.t_1h_2$} (2)
 	edge [loop below left] node[left] {$\frac{1}{4}.t_1t_2$} (2);
 \end{tikzpicture}
 \begin{tikzpicture}[->,>=stealth',shorten >=1pt,auto,node distance=2cm,
                 thick,main node/.style={circle,draw,font=\Large\bfseries}]
 
   \node[main node, initial] (0) {0};
   \node[main node] (1) [below left of=0] {1};
   \node[main node] (2) [below right of=0] {2};
   \node[main node, accepting] (3) [below of=1] {3};
   \node[main node, accepting] (4) [below of=2] {4};
   \node[below right of=3] {{\scriptsize (b) the cryptographer $i$}}; 
 
   \path
     (0) edge [bend left] node [below] {$\frac{1}{4}.hh_i$} (1)
         edge [bend right] node [left] {$\frac{1}{4}.tt_i$} (1)
 	edge [bend left] node {$\frac{1}{4}.ht_i$} (2)
         edge [bend right] node[above] {$\frac{1}{4}.th_i$} (2)
     (1) edge node [left]{$\frac{1}{2}.\neg p_i$} (3)
 	edge node [left]{$\frac{1}{2}.p_i$} (4)
     (2) edge node [right]{$\frac{1}{2}.p_i$} (3)
 	edge node [right]{$\frac{1}{2}.\neg p_i$} (4)
     (3) edge [loop below] node {disagree} (3)
     (4) edge [loop below] node {agree} (4);
 \end{tikzpicture}
 \begin{tikzpicture}[->,>=stealth',shorten >=1pt,auto,node distance=2cm,
                 thick,main node/.style={circle,draw,font=\Large\bfseries}]
 
   \node[main node, initial] (0) {0};
   \node[main node] (1) [below left of=0] {1};
   \node[main node] (2) [below right of=0] {2};
   \node[main node] (3) [below of=1] {3};
   \node[main node] (4) [below of=2] {4};
   \node[main node, accepting] (5) [below of=3] {5};
   \node[main node, accepting] (6) [below of=4] {6};
   \node[below right of=5] {{\scriptsize (c) difference counter}}; 
 
   \path
     (0) edge [bend right] node [left] {$\frac{1}{2}.$agree$_1$} (1)
         edge [bend left] node [right] {$\frac{1}{2}.$disagree$_1$} (2)
     (1) edge [bend right] node [left] {$\frac{1}{2}.$agree$_2$} (3)
         edge node [right] {$\frac{1}{2}.$disagree$_2$} (4)
     (2) edge node [left] {$\frac{1}{2}.$disagree$_2$} (3)
         edge [bend left] node [right] {$\frac{1}{2}.$agree$_2$} (4)
     (3) edge [bend right] node [left] {$\frac{1}{2}.$agree$_3$} (5)
         edge node [right] {$\frac{1}{2}.$disagree$_3$} (6)
     (4) edge node [left] {$\frac{1}{2}.$disagree$_3$} (5)
         edge [bend left] node [right] {$\frac{1}{2}.$agree$_3$} (6)
     (5) edge [loop below] node {1.even} (5)
     (6) edge [loop below] node {1.odd} (6);
 \end{tikzpicture}
 }
 \caption{The model of Dining Cryptographer.}
 \label{fig:dc}
 \end{center}
 \end{figure}